\documentclass[aps,prx,twocolumn,showpacs]{revtex4}

\usepackage{amssymb}

\usepackage{amsthm}
\usepackage{amsmath}
\usepackage{amsfonts}
\usepackage{graphicx}
\usepackage{caption}
\usepackage{mathrsfs}
\usepackage{algorithm}
\usepackage{algpseudocode}
\algtext*{EndFor}
\algtext*{EndWhile}
\algtext*{EndIf}

\newcommand{\ket}[1]{\left|#1\right\rangle}
\newcommand{\bra}[1]{\left\langle#1\right|}

\theoremstyle{plain}
\newtheorem{thm}{Theorem}

\theoremstyle{definition}
\newtheorem{defn}{Definition} 

\newcounter{eqn}
\renewcommand*{\theeqn}{\alph{eqn})}
\newcommand{\num}{\refstepcounter{eqn}\text{\theeqn}\;}

\makeatletter
\newcommand{\putindeepbox}[2][0.7\baselineskip]{{%
    \setbox0=\hbox{#2}%
    \setbox0=\vbox{\noindent\hsize=\wd0\unhbox0}
    \@tempdima=\dp0
    \advance\@tempdima by \ht0
    \advance\@tempdima by -#1\relax
    \dp0=\@tempdima
    \ht0=#1\relax
    \box0
}}
\makeatother

\begin{document}



\title{Quantum Inference on Bayesian Networks}


\author{Guang Hao Low, Theodore J. Yoder, Isaac L. Chuang}
\address{Massachusetts Institute of Technology, 77 Massachusetts Avenue, Cambridge, 02139 MA, United States of America}
\date{\today}


\begin{abstract}
Performing exact inference on Bayesian networks is known to be \#P-hard. Typically approximate inference techniques are used instead to sample from the distribution on query variables given the values $e$ of evidence variables. Classically, a single unbiased sample is obtained from a Bayesian network on $n$ variables with at most $m$ parents per node in time $\mathcal{O}(nmP(e)^{-1})$, depending critically on $P(e)$, the probability the evidence might occur in the first place. By implementing a quantum version of rejection sampling, we obtain a square-root speedup, taking $\mathcal{O}(n2^mP(e)^{-\frac12})$ time per sample. We exploit the Bayesian network's graph structure to efficiently construct a quantum state, a q-sample, representing the intended classical distribution, and also to efficiently apply amplitude amplification, the source of our speedup. Thus, our speedup is notable as it is unrelativized -- we count primitive operations and require no blackbox oracle queries.
\end{abstract}

\maketitle

\section{Introduction}
\label{Introduction}

How are rational decisions made? Given a set of possible actions, the logical answer is the one with the largest corresponding utility. However, estimating these utilities accurately is the problem. A rational agent endowed with a model and partial information of the world must be able to evaluate the probabilities of various outcomes, and such is often done through inference on a Bayesian network \cite{[Russell2003]}, which efficiently encodes joint probability distributions in a directed acyclic graph of conditional probability tables. In fact, the standard model of a decision-making agent in a probabilistic time-discretized world, known as a Partially Observable Markov Decision Process, is a special case of a Bayesian network. Furthermore, Bayesian inference finds application in processes as diverse as system modeling \cite{[Bensi2013]}, model learning \cite{[Neapolitan2004],[Cooper1992]}, data analysis \cite{[Friedman1999]}, and decision making \cite{[Jensen2001]}, all falling under the umbrella of machine learning \cite{[Russell2003]}.

Unfortunately, despite the vast space of applications, Bayesian inference is difficult. To begin with, exact inference is $\#P$ hard in general \cite{[Russell2003]}. It is often far more feasible to perform approximate inference by sampling, such as with the Metropolis-Hastings algorithm \cite{[Metropolis1953]} and its innumerable specializations \cite{[Chib1995]}, but doing so is still NP-hard in general\cite{[Dagum1993]}. This can be understood by considering rejection sampling, a primitive operation common to many approximate algorithms that generates unbiased samples from a target distribution $P(\mathcal{Q}|\mathcal{E})$ for some set of query variables $\mathcal{Q}$ conditional on some assignment of evidence variables $\mathcal{E}=e$. In the general case, rejection sampling requires sampling from the full joint $P(\mathcal{Q},\mathcal{E})$ and throwing away samples with incorrect evidence. In the specific case in which the joint distribution is described by a Bayesian network with $n$ nodes each with no more than $m$ parents, it takes time $\mathcal{O}(nm)$ to generate a sample from the joint distribution, and so a sample from the conditional distribution $P(\mathcal{Q}|\mathcal{E})$ takes average time $\mathcal{O}(n m P(e)^{-1})$. Much of the computational difficulty is related to how the marginal $P(e)= P(\mathcal{E}=e)$ becomes exponentially small as the number of evidence variables increases, since only samples with the correct evidence assignments are recorded.

One very intriguing direction for speeding up approximate inference is in developing hardware implementations of sampling algorithms, for which promising results such as natively probabilistic computing with stochastic logic gates have been reported \cite{[Mansinghka2009]}. In this same vein, we could also consider physical systems that already describe probabilities and their evolution in a natural fashion to discover whether such systems would offer similar benefits.

Quantum mechanics can in fact describe such naturally probabilistic systems. Consider an analogy: if a quantum state is like a classical probability distribution, then measuring it should be analogous to sampling, and unitary operators should be analogous to stochastic updates. Though this analogy is qualitatively true and appealing, it is inexact in ways yet to be fully understood. Indeed, it is a widely held belief that quantum computers offer a strictly {\it more} powerful set of tools than classical computers, even probabilistic ones \cite{[Bernstein1993]}, though this appears difficult to prove \cite{[Aaronson2010]}. Notable examples of the power of quantum computation include exponential speedups for finding prime factors with Shor's algorithm \cite{[Shor1997]}, and square-root speedups for generic classes of search problems through Grover's algorithm \cite{[Grover1996]}. Unsurprisingly, there is a ongoing search for ever more problems amenable to quantum attack \cite{[Galindo2002],[Nielsen2004],[Jordan2013]}. 


For instance, the quantum rejection sampling algorithm for approximate inference was only developed quite recently \cite{[Ozols2012]}, alongside a proof, relativized by an oracle, of a square-root speedup in runtime over the classical algorithm. The algorithm, just like its classical counterpart, is an extremely general method of doing approximate inference, requiring preparation of a quantum pure state representing the joint distribution $P(\mathcal{Q},\mathcal{E})$ and amplitude amplification to amplify the part of the superposition with the correct evidence. Owing to its generality, the procedure assumes access to a state-preparation oracle $\hat{A}_P$, and the runtime is therefore measured by the query complexity \cite{[Grover2000]}, the number of times the oracle must be used. Unsurprisingly, such oracles may not be efficiently implemented in general, as the ability to prepare arbitrary states allows for witness generation to QMA-complete problems \cite{[Ozols2012],[Bookatz2012]}. This also corresponds consistently to the NP-hardness of classical sampling.

In this paper, we present an unrelativized ({\it i.e.} no oracle) square-root, quantum speedup to rejection sampling on a Bayesian network. Just as the graphical structure of a Bayesian network speeds up classical sampling, we find that the same structure allows us to construct the state-preparation oracle $\hat{A}_P$ efficiently. Specifically, quantum sampling from $P(\mathcal{Q}|\mathcal{E}=e)$ takes time $\mathcal{O}(n 2^m P(e)^{-1/2})$, compared with $O(nm P(e)^{-1})$ for classical sampling, where $m$ is the maximum indegree of the network. We exploit the structure of the Bayesian network to construct an efficient quantum circuit $\hat{A}_P$ composed of $\mathcal{O}(n 2^m)$ controlled-NOT gates and single-qubit rotations that generates the quantum state $\ket{\psi_P}$ representing the joint $P(\mathcal{Q},\mathcal{E})$. This state must then be evolved to $\ket{\mathcal{Q}}$ representing $P(\mathcal{Q}|\mathcal{E}=e)$, which can be done by performing amplitude amplification \cite{[Brassard2002]}, the source of our speedup and heart of quantum rejection sampling in general \cite{[Ozols2012]}. The desired sample is then obtained in a single measurement of $\ket{\mathcal{Q}}$.

We better define the problem of approximate inference with a review of Bayesian networks in section~\ref{Bayesian Networks}. We discuss the sensible encoding of a probability distribution in a quantum state axiomatically in section~\ref{Quantum State Encoding}. This is followed by an overview of amplitude amplification in section~\ref{Amplitude Amplification}. The quantum rejection sampling algorithm is given in section \ref{Quantum Rejection Sampling}. As our main result, we construct circuits for the state preparation operator in sections~\ref{Q-sample Preparation} and~\ref{Bayesian State Preparation} and circuits for the reflection operators for amplitude amplification in section~\ref{Phase Flip Operators}. The total time complexity of quantum rejection sampling in Bayesian networks is evaluated in section~\ref{Time Complexity}, and we present avenues for further work in section~\ref{Conclusion}.




\section{Bayesian Networks}
\label{Bayesian Networks}

A Bayesian network is a directed acyclic graph structure that represents a joint probability distribution over $n$ bits. A significant advantage of the Bayesian network representation is that the space complexity of the representation can be made much less than the general case by exploiting conditional dependencies in the distribution. This is achieved by associating with each graph node a conditional probability table for each random variable, with directed edges representing conditional dependencies, such as in Fig.~\ref{DAG1}a.

We adopt the standard convention of capital letters (e.g. $X$) representing random \emph{variables} while lowercase letters (e.g. $a$) are particular fixed \emph{values} of those variables. For simplicity, the random variables are taken to be binary. Accordingly, probability vectors are denoted $P(X)=\left\{P(X=0),P(X=1)\right\}$ while $P(x)\equiv P(X=x)$. Script letters represent a set of random variables $\mathcal{X}=\{X_{1},X_{2},...X_{n}\}$.

An arbitrary joint probability distribution $P(x_1,x_2,\dots,x_n)$ on $n$ bits can always be factored by recursive application of Bayes' rule $P(X,Y)=P(X)P(Y|X)$,
\begin{equation}\label{Bayes_decomp}
P(x_1,x_2,\dots,x_n)=P(x_1)\prod_{i=2}^nP(x_i|x_1,\dots,x_{i-1}).
\end{equation}
However, in most practical situations a given variable $X_i$ will be dependent on only a few of its predecessors' values, those we denote by $\operatorname{parents}(X_i)\subseteq\{x_1,\dots,x_{i-1}\}$ (see Fig.~\ref{DAG1}a). Therefore, the factorization above can be simplified to
\begin{equation}\label{Bayes_Joint}
P(x_1,x_2,\dots,x_n)=P(x_1)\prod_{i=2}^nP(x_i|\operatorname{parents}(X_i)).
\end{equation}
A Bayes net diagrammatically expresses this simplification, with a topological ordering on the nodes $X_1\preceq X_2\preceq\dots\preceq X_n$ in which parents are listed before their children. With a node $x_i$ in the Bayes net, the conditional probability factor $P(x_i=1|\operatorname{parents}(X_i))$ is stored as a table of $2^{m_i}$ values \cite{[Russell2003]} where $m_i$ is the number of parents of node $X_i$, also known as the indegree. Letting $m$ denote the largest $m_i$, the Bayes net data structure stores at most $\mathcal{O}(n2^m)$ probabilities, a significant improvement over the direct approach of storing $\mathcal{O}(2^n)$ probabilities \cite{[Russell2003]}.

\begin{figure*}
\begin{tabular}{cc}
\num\putindeepbox[7pt]{\includegraphics[width=0.24\textwidth]{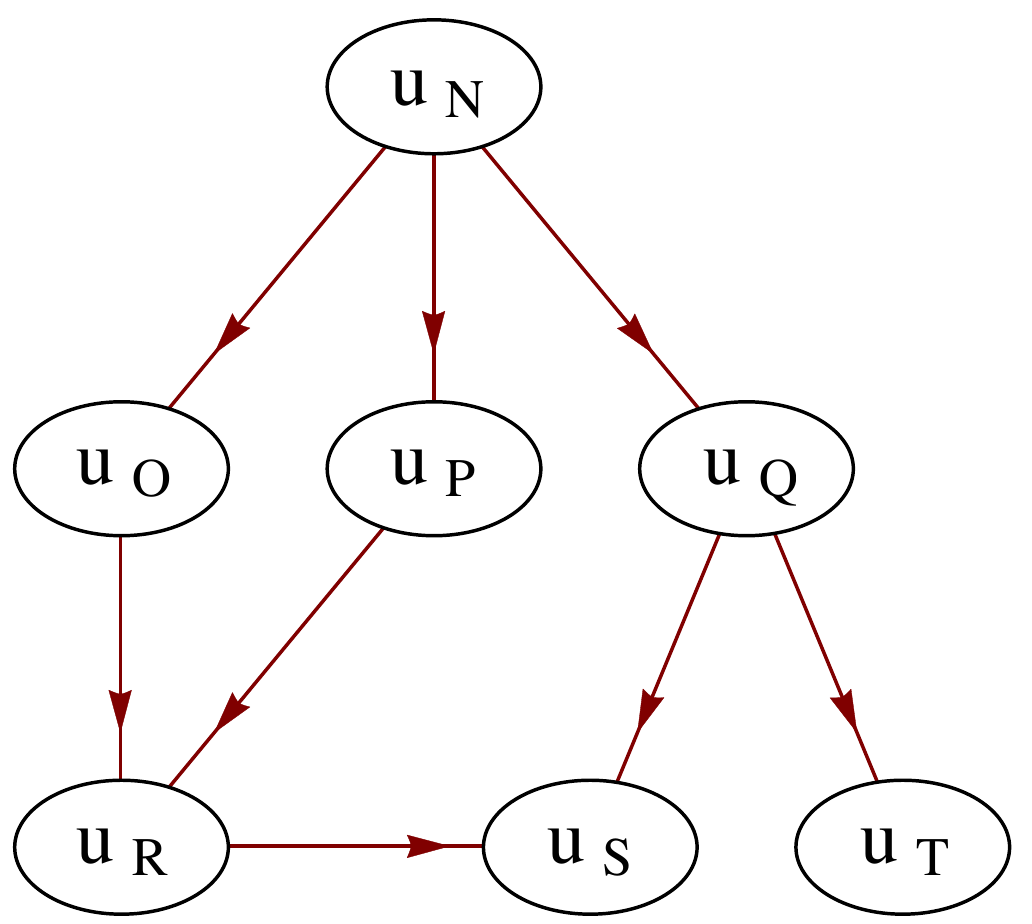}}
 & \num\putindeepbox[7pt]{\includegraphics[width=0.66\textwidth]{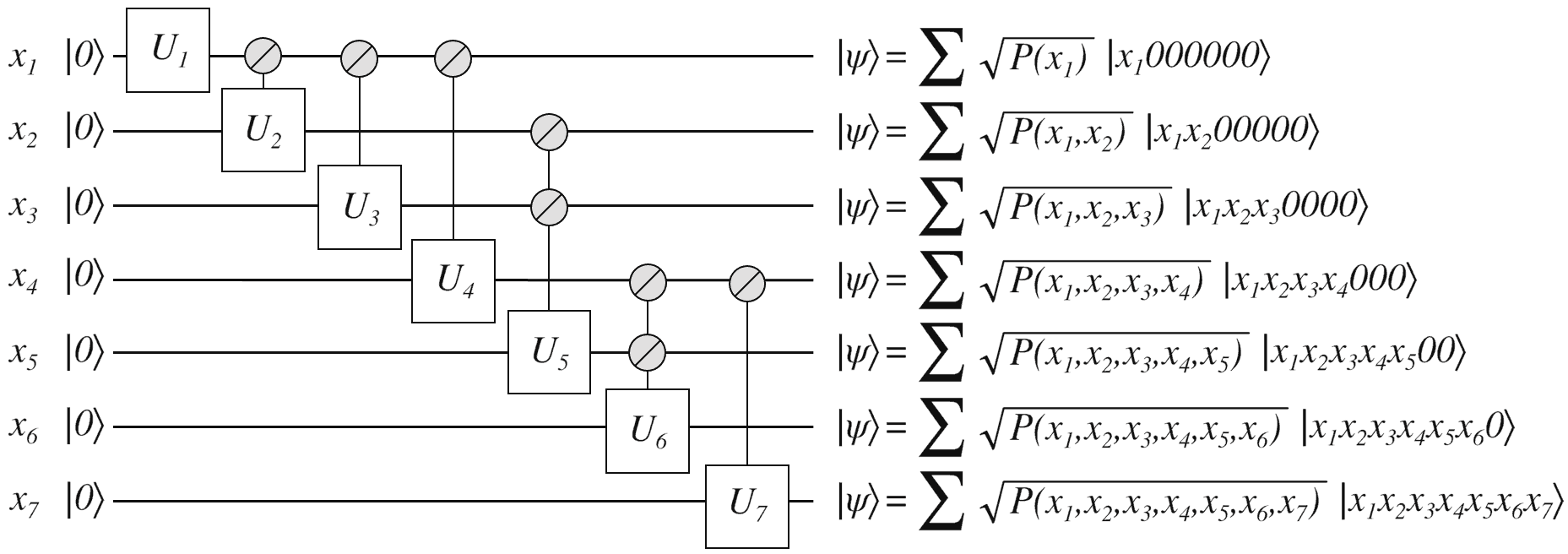}}
 \end{tabular}
\caption{\label{DAG1} a) An example of a directed acyclic graph which can represent a Bayesian network by associating with each node a conditional probability table. For instance, associated with the node $X_1$ is the one value $P(X_1=1)$, while that of $X_5$ consists of four values, the probabilities $X_5=1$ given each setting of the parent nodes, $X_2$ and $X_3$. b) A quantum circuit that efficiently prepares the q-sample representing the full joint distribution of (a). Notice in particular how the edges in the graph are mapped to conditioning nodes in the circuit. The $\ket{\psi_j}$ represent the state of the system after applying the operator sequence $U_1...U_j$ to the initial state $\ket{0000000}$.}
\end{figure*}

A common problem with any probability distribution is inference. Say we have a complete joint probability distribution on $n$-bits, $P(\mathcal{X})$. Given the values $e=e_{|\mathcal{E}|}...e_2 e_1$ for a set $\mathcal{E}\subseteq\mathcal{X}$ of random variables, the task is to find the distribution over a collection of query variables $\mathcal{Q}\subseteq\mathcal{X}\setminus\mathcal{E}$. That is, the exact inference problem is to calculate $P(\mathcal{Q}|\mathcal{E}=e)$. Exact inference is \#P-hard \cite{[Russell2003]}, since one can create a Bayes net encoding the $n$ variable $k$-SAT problem, with nodes for each variable, each clause, and the final verdict -- a count of the satisfying assignments.

Approximate inference on a Bayesian network is much simpler, thanks to the graphical structure. The procedure for sampling is as follows: working from the top of the network, generate a value for each node, given the values already generated for its parents. Since each node has at most $m$ parents that we must inspect before generating a value, and there are $n$ nodes in the tree, obtaining a sample $\{x_1,x_2\dots, x_n\}$ takes time $\mathcal{O}(nm)$. Yet we must postselect on the correct evidence values $\mathcal{E}=e$, leaving us with an average time per sample of $\mathcal{O}\left(nmP(e)^{-1}\right)$, which suffers when the probability $P(e)$ becomes small, typically exponentially small with the number of evidence variables $|\mathcal{E}|$. Quantum rejection sampling, however, will improve the factor of $P(e)^{-1}$ to $P(e)^{-1/2}$, while preserving the linear scaling in the number of variables $n$, given that we use an appropriate quantum state to represent the Bayesian network.



\section{Quantum Sampling from $P(\mathcal{X})$}
\label{Quantum State Encoding}
This section explores the analogy between quantum states and classical probability distributions from first principles. In particular, for a classical probability distribution function $P(\mathcal{X})$ on a set of $n$ binary random variables $\mathcal{X}$ what quantum state $\rho_P$ (possibly mixed, $d$-qubits) should we use to represent it? The suitable state, which we call a quantum probability distribution function (qpdf), is defined with three properties.
\begin{defn}
A qpdf for the probability distribution $P(\mathcal{X})$ has the following three properties:
\begin{enumerate}
\item Purity: In the interest of implementing quantum algorithms, we require the qpdf be a pure state $\rho_P=\ket{\Psi_P}\bra{\Psi_P}$.
\item Q-sampling: A single qpdf can be measured to obtain a classical $n$-bit string, a sample from $P(\mathcal{X})$. Furthermore, for any subset of variables $\mathcal{W}\subset\mathcal{X}$, a subset of qubits in the qpdf can be measured to obtain a sample from the marginal distribution $P(\mathcal{W})$. We call these measurement procedures q-sampling.
\item Q-stochasticity: For every stochastic matrix $T$ there is a unitary $U_T$ such that whenever $T$ maps the classical distribution $P(\mathcal{X})$ to $P'(\mathcal{X})$, $U_T$ maps the qpdf $\ket{\Psi_P}$ to $\ket{\Psi_{P'}}=U_T\ket{\Psi_P}$.
\end{enumerate}
\end{defn}
The motivation for property 3 is for implementing Markov chains, Markov decision processes, or even sampling algorithms such as Metropolis-Hastings, on quantum states. The question we pose, and leave open, is whether a qpdf exists.

The simplest way to satisfy the first two criteria, but not the third, is to initialize a single qubit for each classical binary random variable. This leads to what is called the q-sample, defined in prior work \cite{[Aharonov2003]} as:
\begin{defn}
\label{q-sample-def}
The q-sample of the joint distribution $P(x_1,...,x_n)$ over $n$ binary variables $\{X_i\}$ is the $n$-qubit pure state $\left|\psi_P\right>=\sum_{x_1,...,x_n}\sqrt{P(x_1,...,x_n)}\left|x_1...x_n\right>$.
\end{defn}
The q-sample possesses property 1 and the eponymous property 2 above. However, it does not allow for stochastic updates as per property 3, as a simple single qubit example shows. In that case, property 3 requires
\begin{equation}\label{SUC}
\left(\begin{array}{cc}U_{11}&U_{12}\\U_{21}&U_{22}\end{array}\right)\left(\begin{array}{c}\sqrt{p}\\\sqrt{1-p}\end{array}\right)=\left(\begin{array}{c}\sqrt{p T_{11}+(1-p)T_{12}}\\\sqrt{pT_{21}+(1-p)T_{22}}\end{array}\right),
\end{equation}
for all $p\in[0,1]$. Looking at Eq.~\ref{SUC} for $p=0$ and $p=1$ constrains $U$ completely, and it is never unitary when $T$ is stochastic. Thus, the q-sample fails to satisfy property 3.

Yet, the q-sample satisfies properties 1 and 2 in a very simplistic fashion, and various more complicated forms might be considered. For instance, relative phases could be added to the q-sample giving $\sum_x e^{i\phi(x)}\sqrt{P(x)}\ket{x}$, though this alone does not guarantee property 3, which is easily checked by adding phases to the proof above. Other extensions of the q-sample may include ancilla qubits, different measurement bases, or a post-processing step including classical randomness to translate the measurement result into a classical sample. It is an open question whether a more complicated representation satisfying all three properties exists, including q-stochasticity, so that we would have a qpdf possessing all the defining properties.

Nevertheless, although the q-sample is not a qpdf by our criteria, it will still be very useful for performing quantum rejection sampling. The property that a sample from a marginal distribution is obtained by simply measuring a subset of qubits means that, using conditional gates, we can form a q-sample for a conditional distribution from the q-sample for the full joint distribution, as we will show in section \ref{Circuit Constructions}. This corresponds to the classical formula $P(\mathcal{V}|\mathcal{W})=P(\mathcal{V},\mathcal{W})/P(\mathcal{W})$, which is the basis behind rejection sampling. The way it is actually done quickly on a q-sample is through amplitude amplification, reviewed next, in the general case.

\section{Amplitude Amplification}
\label{Amplitude Amplification}
Amplitude amplification \cite{[Brassard2002]} is a well-known extension of Grover's algorithm and is the second major concept in the quantum inference algorithm. Given a quantum circuit $\hat{A}$ for the creation of an $n$-qubit pure state $\ket{\psi}=\alpha\ket{\psi_t}+\beta\ket{\bar\psi_t}=\hat{A}\ket{0}^{\otimes n}$, where $\langle\psi_t|\bar\psi_t\rangle=0$, the goal is to return the target state $\ket{\psi_t}$ with high probability. To make our circuit constructions more explicit, we assume target states are marked by a known evidence bit string $e=e_{|\mathcal{E}|}...e_2e_1$, so that $\ket{\psi_t}=\ket{\mathcal{Q}}\ket{e}$ lives in the tensor product space $\mathcal{H}_\mathcal{Q}\otimes\mathcal{H}_\mathcal{E}$ and the goal is to extract $\ket{\mathcal{Q}}$.

Just like in Grover's algorithm, a pair of reflection operators are applied repetitively to rotate $\ket{\psi}$ into $\ket{\psi_t}$. Reflection about the evidence is performed by $\hat S_e=\hat I\otimes(\hat I-2\ket{e}\bra{e})$ followed by reflection about the initial state by $\hat S_\psi=(\hat I-2\ket{\psi}\bra{\psi})$. Given $\hat A$, then $\hat S_\psi=\hat A\hat S_0\hat A^\dagger$, where $\hat S_0=(\hat I-2\ket{0}\bra{0}^{\otimes n})$.

The analysis of the amplitude amplification algorithm is elucidated by writing the Grover iterate $\hat G=-\hat S_\psi\hat S_e=-\hat A\hat
S_0\hat A^\dagger\hat S_e$ in the basis of 
$\frac{\alpha}{|\alpha|}\ket{\psi_t}\equiv\left(\begin{smallmatrix}1\\0\end{smallmatrix}\right)$ and 
$\frac{\beta}{|\beta|}\ket{\bar\psi_t}\equiv\left(\begin{smallmatrix}0\\1\end{smallmatrix}\right)$ \cite{[Grover2000]},
\begin{equation}
\hat 
G=\left(\begin{array}{cc}1-2|\alpha|^2&2|\alpha|\sqrt{1-|\alpha|^2}\\-2|\alpha|\sqrt{1-|\alpha|^2}&1-2|\alpha|^2\end{array}\right).
\end{equation}
In this basis, the Grover iterate corresponds to a rotation by small 
angle $\theta=\cos^{-1}(1-2|\alpha|^2)\approx2|\alpha|$. Therefore, 
applying the iterate $N$ times rotates the state by $N\theta$. We conclude that 
$\hat G^N\ket{\psi}$ is closest to $\frac{\alpha}{|\alpha|}\ket{\psi_t}$ 
after $N=\mathcal{O}\left(\frac{\pi}{4|\alpha|}\right)$ iterations.

Usually, amplitude amplification needs to be used without knowing the value of $|\alpha|$. In that case, $N$ is not known. However, the situation is remedied by guessing the correct number of Grover iterates to apply in exponential progression. That is, we apply $\hat G$ $2^k$ times, with $k=0,1,2,\dots$, measure the evidence qubits $\ket{\mathcal{E}}$ after each attempt, and stop when we find $\mathcal{E}=e$. It has been shown \cite{[Brassard2002]} that this approach also requires on average $\mathcal{O}\left(\frac{1}{|\alpha|}\right)$ applications of $\hat G$.

\section{The Quantum Rejection Sampling Algorithm}
\label{Quantum Rejection Sampling}

The quantum rejection sampling algorithm \cite{[Ozols2012]}, which we review now, is an application of amplitude amplification on a q-sample. The general problem, as detailed in section \ref{Bayesian Networks}, is to sample from the $n$-bit distribution $P(\mathcal{Q}|\mathcal{E}=e)$. We assume that we have a circuit $\hat A_P$ that can prepare the q-sample $\ket{\psi_P}=\hat A_P\ket{0}^{\otimes n}$. Now, permuting qubits so the evidence lies to the right, the q-sample can be decomposed into a superposition of states with correct evidence and states with incorrect evidence.
\begin{equation}
\label{evid_decomp}
\ket{\psi_P}=\sqrt{P(e)}\ket{\mathcal{Q}}\ket{e}+\sqrt{1-P(e)}\ket{\overline{\mathcal{Q},e}},
\end{equation}
where $\ket{\mathcal{Q}}$ denotes the q-sample of $P(\mathcal{Q}|\mathcal{E}=e)$, our target state. Next perform the amplitude amplification algorithm from the last section to obtain $\ket{\mathcal{Q}}$ with high probability. Note that this means the state preparation operator $\hat A_P$ must be applied $\mathcal{O}(P(e)^{-1/2})$ times. Once obtained, $\ket{\mathcal{Q}}$ can be measured to get a sample from $P(\mathcal{Q}|\mathcal{E}=e)$, and we have therefore done approximate inference. Pseudocode is provided as an algorithm \ref{alg1}.

However, we are so far missing a crucial element. How is the q-sample preparation circuit $\hat A_P$ actually implemented, and can this implementation be made efficient, that is, polynomial in the number of qubits $n$? The answer to this question removes the image of $\hat A_P$ as a featureless black box and is addressed in the next section.

\begin{algorithm}                      
\caption{\label{alg1} Quantum rejection sampling algorithm: generate one sample from $P(\mathcal{Q}|\mathcal{E}=e)$ given a q-sample preparation circuit $\hat A_P$}          
\begin{algorithmic}
\State $k \leftarrow -1$
\While {evidence $\mathcal{E}\neq e$}
\State $k \leftarrow k+1$
\State $\ket{\psi_P} \leftarrow \hat A_P\ket{0}^{\otimes n}$ \space\//\//prepare a q-sample of $P(\mathcal{X})$
\State $\ket{\psi_P'} \leftarrow \hat G^{2^k}\ket{\psi_P}$ \space\//\//where $\hat G=-\hat{A}_P\hat{S}_0 \hat{A}_P^\dag \hat{S}_e$
\State Measure evidence qubits $\mathcal{E}$ of $\ket{\psi_P'}$
\EndWhile
\State Measure the query qubits to obtain a sample $\mathcal{Q}=q$
\end{algorithmic}                          
\end{algorithm}


\section{Circuit Constructions}
\label{Circuit Constructions}

While the rejection sampling algorithm from Section \ref{Quantum Rejection Sampling} is entirely general for any distribution $P(\mathcal{X})$, the complexity of q-sample preparation, in terms of the total number of CNOTs and single qubit rotations involved, is generally exponential in the number of qubits, $\mathcal{O}(2^n)$. We show this in section \ref{Q-sample Preparation}. The difficultly is not surprising, since arbitrary q-sample preparation encompasses witness generation to QMA-complete problems \cite{[Ozols2012],[Bookatz2012]}. However, there are cases in which the q-sample can be prepared efficiently \cite{[Aharonov2003]}. The main result of this paper is that, for probability distributions resulting from a Bayesian network $\mathcal{B}$ with $n$ nodes and maximum indegree $m$, the circuit complexity of the q-sample preparation circuit $\hat A_{\mathcal{B}}$ is $\mathcal{O}\left(n2^m\right)$. We show this in section \ref{Bayesian State Preparation}. The circuit constructions for the remaining parts of the Grover iterate, the phase flip operators, are given in section \ref{Phase Flip Operators}. Finally, we evaluate the complexity of our constructions as a whole in section \ref{Time Complexity} and find that approximate inference on Bayesian networks can be done with a polynomially sized quantum circuit.

Throughout this section we will denote the circuit complexity of a circuit $\hat C$ as $Q_{\hat C}$. This complexity measure is the count of the number of gates in $\hat C$ after compilation into a complete, primitive set. The primitive set we employ includes the CNOT gate and all single qubit rotations.

\subsection{Q-sample Preparation}
\label{Q-sample Preparation}

If $P(x)$ lacks any kind of structure, the difficulty of preparing the q-sample $\left|\psi_P\right>=\hat{A}_P\left|0\right>^{\otimes n}$ with some unitary $\hat{A}_P$  scales at least exponentially with the number of qubits $n$ in the q-sample. Since $P(x)$ contains $2^n-1$ arbitrary probabilities, $\hat A_P$ must contain at least that many primitive operations. In fact, the bound is tight --- we can construct a quantum circuit preparing $\ket{\psi_P}$ with complexity $\mathcal{O}(2^n)$.

\begin{thm}
\label{arbitrary_state_preparation}
Given an arbitrary joint probability distribution $P(x_1,...,x_n)$ over $n$ binary variables $\{X_i\}$, there exists a quantum circuit $\hat{A}_P$ that prepares the q-sample $\hat{A}_P\left|0\right>^{\otimes n}=\ket{\psi_P}=\sum_{x_1,...,x_n}\sqrt{P(x_1,...,x_n)}\left|x_1...x_n\right>$ with circuit complexity $\mathcal{O}(2^n)$.
\end{thm}
\begin{proof}
Decompose ${P(x)=P(x_{ 1 })\prod _{ i=2 }^{ n } P(x_{ i }|x_{ 1 }...x_{ i-1 })}$ as per Eq.~\eqref{Bayes_decomp}. For each conditional distribution $P(X_{ i }|x_{ 1 }...x_{ i-1 })$, let us define the $i$-qubit uniformly controlled rotation $\hat{U}_i$ such that given an $(i-1)$ bit string assignment $x_c\equiv x_{ 1 }...x_{ i-1 }$ on the control qubits, the action of $\hat{U}_{i}$ on the $i^{\text{th}}$ qubit initialized to $\left|0\right>_{i}$ is a rotation about the y-axis by angle $2 \tan^{-1}(\sqrt{P(x_{i}=1|x_c)/P(x_{i}=0|x_c)})$ or $\hat{U}_{i}\left|0\right>_{i}=\sqrt{P(x_{i}=0|x_c)}\left|0\right>_{i}+\sqrt{P(x_{i}=1|x_c)}\left|1\right>_{i}$. With this definition, the action of the single-qubit $\hat{U}_{1}$ is $\hat{U}_{1}\left|0\right>_{1}=\sqrt{P(x_{1}=0)}\left|0\right>_{1}+\sqrt{P(x_{1}=1)}\left|1\right>_{1}$. By applying Bayes' rule in reverse, the operation $\hat{A}_P=\hat{U}_{n}...\hat{U}_{1}$ then produces $\left|\psi_P\right>=\hat{A}_P\left|0\right>$. As each $k$-qubit uniformly controlled rotation is decomposable into $\mathcal{O}(2^k)$ CNOTs and single-qubit rotations \cite{[Bergholm2005]}, the circuit complexity of $\hat{A}_P$ is $Q_{\hat{A}_P}=\sum_{i=1}^{n}\mathcal{O}(2^{i-1})=\mathcal{O}(2^n)$.
\end{proof}

The key quantum compiling result used in this proof is the construction of Bergholm et.~al.~\cite{[Bergholm2005]} that decomposes $k$-qubit uniformly controlled gates into $\mathcal{O}(2^k)$ CNOTs and single qubit operations. Each uniformly controlled gate is the realization of a conditional probability table from the factorization of the joint distribution. We use this result again in Bayesian q-sample preparation.

\subsection{Bayesian Q-sample Preparation}
\label{Bayesian State Preparation}
We now give our main result, a demonstration that the circuit $\hat A_\mathcal{B}$ that prepares the q-sample of a Bayesian network is exponentially simpler than the general q-sample preparation circuit $\hat A_P$. We begin with a Bayesian network with, as usual, $n$ nodes and maximum indegree $m$ that encodes a distribution $P(\mathcal{X})$. As a minor issue, because the Bayesian network may have nodes reordered, the indegree $m$ is actually a function of the specific parentage of nodes in the tree. This non-uniqueness of $m$ corresponds to the non-uniqueness of the decomposition $P(x_1,...,x_n)=P(x_{ 1 })\prod _{ i=2 }^{ n } P(x_{ i }|x_{ 1 }...x_{ i-1 })$ due to permutations of the variables. Finding the variable ordering minimizing $m$ is unfortunately an NP-hard problem \cite{[Chickering2004]}, but typically the variables have real-world meaning and the natural causal ordering often comes close to optimal \cite{[Druzdzel1993]}. In any case, we take $m$ as a constant much less than $n$.


\begin{defn}
\label{q-sample}
If $P(\mathcal{X})$ is the probability distribution represented by a Bayesian network $\mathcal{B}$, the Bayesian q-sample $\ket{\psi_{\mathcal{B}}}$ denotes the q-sample of $P(\mathcal{X})$.
\end{defn}

\begin{thm}
\label{bayesian_state_preparation}
The Bayesian q-sample of the Bayesian network $\mathcal{B}$ with $n$ nodes and bounded indegree $m$ can be prepared efficiently by an operator $\hat{A}_{\mathcal{B}}$ with circuit complexity $\mathcal{O}(n 2^m)$ acting on the initial state $\left|0\right>^{\otimes n}$.
\end{thm}
\begin{proof}
As a Bayesian network is a directed acyclic graph, let us order the node indices topologically such that for all $1\le i\le n$, we have $\text{parents}(x_i)\subseteq\{x_1,x_2,\dots,x_{i-1}\}$, and $\text{max}_i\left|\text{parents}(x_{i})\right|=m$. Referring to the construction from the proof of theorem~\ref{arbitrary_state_preparation}, the state preparation operator $\hat{A}=\hat{U}_{n}...\hat{U}_{1}$ then contains at most $m$-qubit uniformly controlled operators, each with circuit complexity $\mathcal{O}(2^m)$, again from Bergholm et. al. \cite{[Bergholm2005]}. The circuit complexity of $\hat{A}_{\mathcal{B}}$ is thus $Q_{\hat{A}_{\mathcal{B}}}=\sum_{i=1}^{n}\mathcal{O}(2^m)=\mathcal{O}(n2^m)$.
\end{proof}

Fig.~\ref{DAG1}b shows the circuit we have just described. Bayesian q-sample preparation forms part of the Grover iterate required for amplitude amplification. The rest is comprised of the reflection, or phase flip, operators.

\subsection{Phase Flip Operators}
\label{Phase Flip Operators}
Here we show that the phase flip operators are also efficiently implementable, so that we can complete the argument that amplitude amplification on a Bayesian q-sample is polynomial time. Note first that the phase flip operators $\hat S_e$ acting on $k=|\mathcal{E}|\le n$ qubits can be implemented with a single $k$-qubit controlled $Z$ operation along with at most $2k$ bit flips. The operator $\hat S_0$ is the special case $\hat S_{e=0^n}$. A Bayesian q-sample can be decomposed exactly as in Eq.~\eqref{evid_decomp}
\begin{equation}
\ket{\psi_\mathcal{B}}=\sqrt{P(e)}\ket{\mathcal{Q}}\ket{e}+\sqrt{1-P(e)}\ket{\overline{\mathcal{Q},e}}.
\end{equation}
Recall $\ket{\mathcal{Q}}$ is the q-sample of $P(\mathcal{Q}|\mathcal{E}=e)$ and $\ket{\overline{\mathcal{Q}e}}$ contains all states with invalid evidence $\mathcal{E}\neq e$. We write the evidence as a $k$-bit string $e=e_k\dots e_2e_1$ and $\hat X_i$ as the bit flip on the $i^{\text{th}}$ evidence qubit. The controlled phase, denoted $\hat Z_{1\dots k}$, acts on all $k$ evidence qubits symmetrically, flipping the phase if and only if all qubits are 1. Then $\hat S_e$ is implemented by
\begin{equation}
\hat S_e=\hat B\hat Z_{1\dots k} \hat B
\end{equation}
where $\hat B=\prod_{i=1}^k \hat X_i^{\bar e_i}$ with $\bar e_i\equiv1-e_i$. Explicitly,
\begin{align}
\hat S_e\ket{\psi_\mathcal{B}}&=\hat B\hat Z_{1\dots k} \hat B\left[\sqrt{P(e)}\ket{\mathcal{Q}}\ket{e}+\sqrt{1-P(e)}\ket{\overline{\mathcal{Q}e}}\right]\\\nonumber
&=\hat B\hat Z_{1\dots k}\left[\sqrt{P(e)}\ket{\mathcal{Q}}\ket{1^n}+\sqrt{1-P(e)}\ket{\overline{\mathcal{Q}1^n}}\right]\\\nonumber
&=\hat B\left[-\sqrt{P(e)}\ket{\mathcal{Q}}\ket{1^n}+\sqrt{1-P(e)}\ket{\overline{\mathcal{Q}1^n}}\right]\\\nonumber
&=\left[-\sqrt{P(e)}\ket{\mathcal{Q}}\ket{e}+\sqrt{1-P(e)}\ket{\overline{\mathcal{Q}e}}\right].
\end{align}

The circuit diagram representing $\hat S_e$ is shown in Fig.~\ref{phase_flip_circuit}. The $k$-qubit controlled phase can be constructed from $\mathcal{O}(k)$ CNOTs and single qubit operators using $\mathcal{O}(k)$ ancillas \cite{[Nielsen2004]} or, alternatively, $\mathcal{O}(k^2)$ CNOTs and single qubit operators using no ancillas \cite{[Saeedi2013]}.

\begin{figure}[H]
\includegraphics[width=\columnwidth]{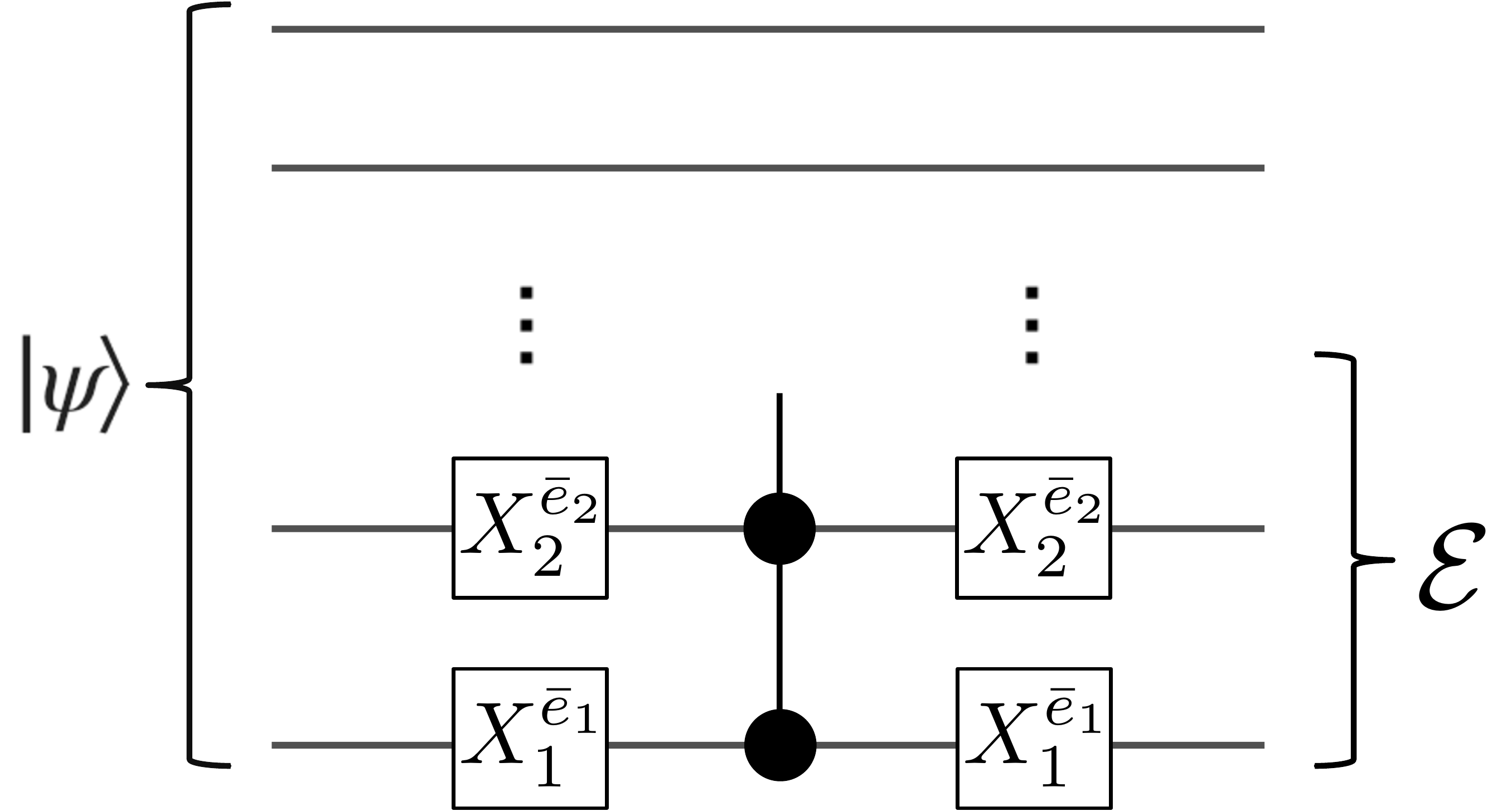}
\caption{\label{phase_flip_circuit}Quantum circuit for implementing the phase flip operator $\hat S_e$. The $k=|\mathcal{E}|$ qubit controlled phase operator acts on the evidence qubits $\mathcal{E}$. It may be compiled into $\mathcal{O}(k)$ CNOTs and single qubit operators given $\mathcal{O}(k)$ ancillas \cite{[Nielsen2004]}. The evidence values $e=e_k\dots e_2e_1$ control the bit flips through $\bar e_i\equiv1-e_i$.}
\end{figure}
\begin{table}[H]
\centering
\begin{tabular}{ l | c | r }
  $\hat{U}$ & $Q_{\hat{U}}$ & \text{Comments} \\ \hline
  $\hat{A}_P$ & $\mathcal{O}(2^n)$ & \text{Q-sample preparation} \\
  $\hat{A}_\mathcal{B}$ & $\mathcal{O}(n 2^m)$ & \text{Bayesian state preparation} \\
  $\hat{S}_0$ & $\mathcal{O}(n)$ & $\mathcal{O}(n)$ ancilla qubits \\
  $\hat{S}_e$ & $\mathcal{O}(|\mathcal{E}|)$ & $\mathcal{O}(|\mathcal{E}|)$ ancilla qubits
\end{tabular}
\caption{\label{Complexity_table} Circuit complexity $Q_{\hat{U}}$ of implementing the operators $\hat U$ discussed in the text. The Grover iterate $\hat G$ for amplitude amplification of a Bayesian q-sample (general q-sample) consists of two instances of the preparation circuit $\hat A_\mathcal{B}$ ($\hat A_P$) and one instance each of $\hat S_0$ and $\hat S_e$. The time to collect one sample from $P(\mathcal{Q}|\mathcal{E}=e)$ is $\mathcal{O}(Q_{\hat G}P(e)^{-1/2})$.}
\end{table}

\subsection{Time Complexity}
\label{Time Complexity}
The circuit complexities of the various elements in the Grover iterate $\hat{G}=-\hat{A}\hat{S}_0 \hat{A}^\dag \hat{S}_e$ are presented in Table~\ref{Complexity_table}. As the circuit complexity of the phase flip operator $S_{0}$ ($S_{e}$) scales linearly with number of qubits $n$ ($|\mathcal{E}|$), $Q_{\hat{G}}$ is dominated by the that of the state preparation operator $\hat A$. Although $Q_{\hat{A}_P}$ scales exponentially with the number of nodes $n$ for general q-sample preparation, Bayesian q-sample preparation on a network of bounded indegree $m$ is efficient. Namely, $Q_{\hat{A}_\mathcal{B}}=\mathcal{O}(n 2^m)$ scales linearly with $n$ as in classical sampling from a Bayesian network. It takes $\mathcal{O}(P(e)^{-1/2})$ applications of $\hat A_{\mathcal{B}}$ to perform the rejection sampling algorithm from Section \ref{Quantum Rejection Sampling} and, thus, a single sample from $P(\mathcal{Q}|\mathcal{E})$ can be obtained by a quantum computer in time $\mathcal{O}(n2^mP(e)^{-1/2})$. In Section \ref{Bayesian Networks}, we saw that classical Bayesian inference takes time $\mathcal{O}(nmP(e)^{-1})$ to generate a single sample. Thus, quantum inference on a Bayesian network provides a square-root speedup over the classical case. The quantum circuit diagram for Bayesian inference is outlined in Fig.~\ref{CompleteCircuit}.

\setcounter{eqn}{0}
\renewcommand{\arraystretch}{2.5}
\begin{figure}
\begin{tabular}{l}
\num\putindeepbox[7pt]{\includegraphics[width=0.9\columnwidth]{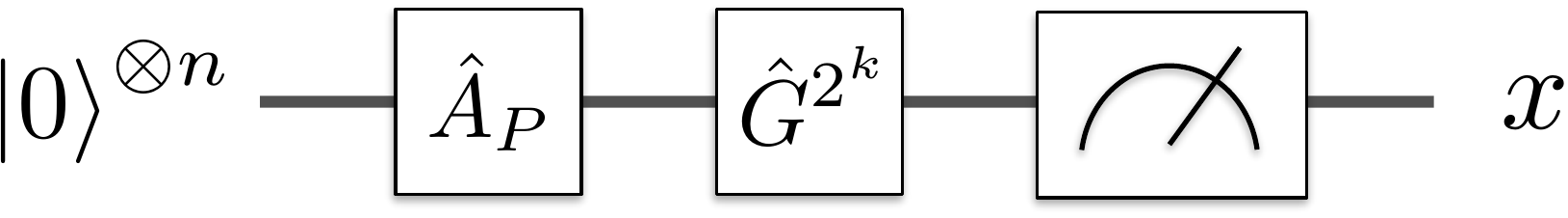}}\\
\num\putindeepbox[7pt]{\includegraphics[width=0.9\columnwidth]{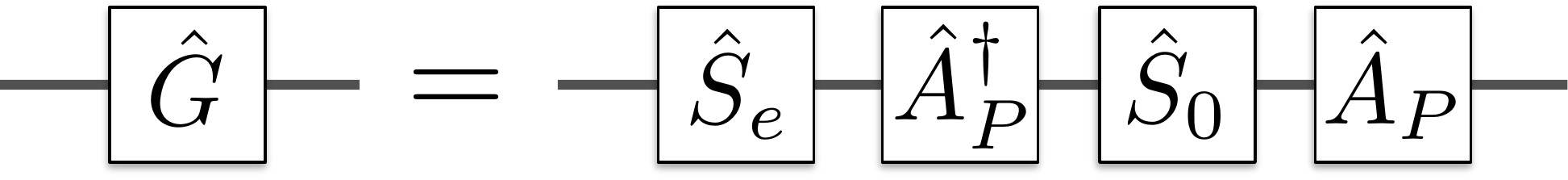}}
\end{tabular}
\captionsetup{justification=centerlast, singlelinecheck=false}
\caption{\label{CompleteCircuit} a) Quantum Bayesian inference on Bayes net $\mathcal{B}$ for evidence $\mathcal{E}=e$ is done by repetition of the circuit shown, with $k$ incrementing $k=0,1,\dots$, stopping when the measurement result $x$ contains evidence bits $e$. Then $x$ can be recorded as a sample from the conditional distribution $P(\mathcal{Q}|\mathcal{E})$. This corresponds to Algorithm \ref{alg1}. b) The constituents of the Grover iterate $\hat G$, the state preparation $\hat A_\mathcal{B}$ and phase flip operators $\hat S_e$ and $\hat S_0$. The state preparation operator is constructed from Theorem 2, and an example is shown in Fig.~\ref{DAG1}b. The phase flip operators are constructed as shown in Fig.~\ref{phase_flip_circuit}.\space\space\space\space\space\space\space\space\space\space\space\space\space\space\space\space\space\space\space\space\space\space\space\space\space\space}
\end{figure}

\section{Conclusion}
\label{Conclusion}
We have shown how the structure of a Bayesian network allows for a square-root, quantum speedup in approximate inference. We explicitly constructed a quantum circuit from CNOT and single qubit rotations that returns a sample from $P(\mathcal{Q}|\mathcal{E}=e)$ using just $\mathcal{O}(n2^mP(e)^{-\frac12})$ gates. For more general probability distributions, the Grover iterate would include a quantity of gates exponential in $n$, the number of random variables, and thus not be efficient. This efficiency of our algorithm implies experimental possibilities. As a proof of principle, one could experimentally perform inference on a two node Bayesian network with only two qubits with current capabilities of ion trap qubits \cite{[Hanneke2010]}.

We also placed the idea of a q-sample into the broader context of an analogy between quantum states and classical probability distributions. If a qpdf can be found that is pure, can be q-sampled, and allows q-stochastic updates, the quantum machine learning subfield would greatly benefit. Algorithms for many important routines, such as Metropolis-Hastings, Gibbs-sampling, and even Bayesian learning, could find square-root speedups in a similar manner to our results here.

Artificial intelligence and machine learning tasks are often at least NP-hard. Although exponential speedups on such problems are precluded by BBBV \cite{[Bennett1997]}, one might hope for square-root speedups, as we have found here, for a variety of tasks. For instance, a common machine learning environment is online or interactive, in which the agent must learn while making decisions. Good algorithms in this case must balance exploration, finding new knowledge, with exploitation, making the best of what is already known. The use of Grover's algorithm in reinforcement learning has been explored \cite{[Dong2008]}, but much remains to be investigated. One complication is that machine learning often takes place in a classical world; a robot is not usually allowed to execute a superposition of actions. One might instead focus on learning tasks that take place in a purely quantum setting. For instance, quantum error correcting codes implicitly gather information on what error occurred in order to correct it. Feeding this information back into the circuit, could create an adaptive, intelligent error correcting code.

\bibliographystyle{elsarticle-num}
\bibliography{QIReferencesMin}







\end{document}